\documentclass[10pt,twocolumn,twoside]{IEEEtran} 

\usepackage{graphicx}
\usepackage{subfigure}
\usepackage{caption}
\usepackage{amsmath}
\usepackage{amssymb}
\usepackage{amsthm} 
\usepackage{accents}
\usepackage{statex}
\usepackage{cases}
\usepackage{mathrsfs}
\usepackage{tikz,pgfplots}
\pgfplotsset{compat=newest}
\usepackage{xcolor}
\usepackage{caption}
\usepackage{diagbox}
\usepackage{setspace}
\usepackage{booktabs}
\usepackage{array}
\usepackage{cite}
\usepackage{enumitem}
\usepackage{algorithm}
\usepackage{algorithmic}
\usepackage{multirow}
\usepackage{mathrsfs}
\usetikzlibrary{patterns}
\usetikzlibrary{graphs}
\usetikzlibrary{graphs.standard}
\usetikzlibrary{matrix,arrows,fit,backgrounds,mindmap,plotmarks,decorations.pathreplacing}
\newcolumntype{C}[1]{>{\centering\arraybackslash}p{#1}}
\usepackage[colorlinks,linkcolor=blue]{hyperref}

\usetikzlibrary{automata}

\newtheorem{proposition}{Proposition}
\newtheorem{theorem}{Theorem}
\newtheorem{definition}{Definition}
\newtheorem{lemma}{Lemma}

\newtheorem{remark}{Remark}

\newtheorem{condition}{Condition}

\newlength\figureheight
\newlength\figurewidth

\DeclareMathOperator{\1}{\textbf{1}} 
 
\DeclareMathOperator*{\argmin}{argmin}

\DeclareMathOperator*{\rank}{rank\,}

\DeclareMathOperator*{\mI}{\mathcal{I}}
\DeclareMathOperator*{\mH}{\mathscr{H}}
\DeclareMathOperator*{\mB}{\mathscr{B}}
\DeclareMathOperator*{\idx}{idx}
\DeclareMathOperator*{\nul}{Null\,}

\newcommand{\revise}[1]{\textcolor{blue}{#1}}

\newcommand{\algorithmicbreak}{\textbf{break}}

\newcommand{\la}[1]{\label{#1}}
\newcommand{\re}[1]{(\ref{#1})}
\newcommand{\bmx}{\begin{bmatrix}}
\newcommand{\emx}{\end{bmatrix}}
\newcommand{\bsm}{\left[\begin{smallmatrix}}
\newcommand{\esm}{\end{smallmatrix}\right]}
\newcommand{\R}{\ensuremath{\mathbb{R}}} 
\newcommand{\calL}{\mathscr{L}} 
\newcommand{\HH}{\mathscr{H}}   
\newcommand{\Beh}{\mathscr{B}}   
\newcommand{\invn}{\bm{n}} 
\newcommand{\invl}{\bm{\ell}} 
\newcommand{\invp}{\bm{p}} 
\newcommand{\invm}{\bm{m}} 

\title{Secure Data Reconstruction:\\ A Direct Data-Driven Approach}

\author{Jiaqi Yan, Ivan Markovsky, and John Lygeros
	\thanks{
		J. Yan is with the School of Automation Science and Electrical Engineering, Beihang University, Beijing 100191, P. R. China. (jqyan@buaa.edu.cn)}
\thanks{I. Markovsky is with the Catalan Institution for Research and Advanced Studies (ICREA), Barcelona, Spain and the International Centre for Numerical Methods in Engineering (CIMNE), Barcelona, Spain. (imarkovsky@cimne.upc.edu)}
\thanks{
 J. Lygeros is with the Automatic Control Laboratory, ETH Zurich. (jiayan@ethz.ch)}
\thanks{This work was supported by the Swiss National Science Foundation through NCCR Automation under Grant agreement 51NF40\_180545, ICREA, and the Fond for Scientific Research Vlaanderen (FWO) project G033822N.}
}

\begin{document}
\maketitle

\begin{abstract}

This paper addresses the problem of secure data reconstruction for unknown systems, where data collected from the system are susceptible to malicious manipulation. We aim to recover the real trajectory without prior knowledge of the system model. To achieve this, a behavioral language is used to represent the system, describing it using input/output trajectories instead of state-space models. We consider two attack scenarios. In the first scenario, up to $k$ entries of the collected data are malicious. On the other hand, the second scenario assumes that at most $k$ channels from sensors or actuators can be compromised, implying that any data collected from these channels might be falsified. For both scenarios, we formulate the trajectory recovery problem as an optimization problem and introduce sufficient conditions to ensure successful recovery of the true data. Since finding exact solutions to these problems can be computationally inefficient, we further approximate them using an $\ell_1$-norm and group Least Absolute Shrinkage and Selection Operator (LASSO). We demonstrate that under certain conditions, these approximation problems also find the true trajectory while maintaining low computation complexity. Finally, we extend the proposed algorithms to noisy data. By reconstructing the secure trajectory, this work serves as a safeguard mechanism for subsequent data-driven control methods.

\end{abstract}


\section{Introduction}
The past decades have witnessed remarkable research interest in cyber-physical systems (CPSs). Application of these systems varies from aerospace, manufacturing, transportation, and power grids. Security stands as a top concern for CPS, since any failure or attack can lead to significant economic losses, or even pose threats to human lives \cite{cardenas2008secure}.
Nowadays, the extensive use of communication units (e.g., sensors and actuators) makes the system more vulnerable than before, where an adversary can easily break a communication channel and tamper with the transmitted data. Researchers, therefore, have recognized the importance of implementing secure detection and control algorithms.

Existing solutions typically rely on model-based approaches. One common strategy is detecting attacks by creating an observer based on the system model \cite{chowdhury2020observer,teixeira2010networked}, \revise{where both linear and nonlinear systems are considered}. One estimates the process states using this observer, and \revise{makes} a decision on possible misbehaviors by comparing these estimates with the received data. An application of this approach is bad data detection in power systems \cite{deng2016false}. 

Instead of detecting and isolating malicious components, another line of research focuses on developing resilient algorithms which guarantee an acceptable performance even in the presence of false data. This is usually achieved by deploying additional nodes or channels to ``compensate" for fault signals. For instance, \cite{mishra2016secure,mao2022computational,li2021low,nakahira2018attack} solve the problem of secure state estimation in \revise{linear systems}. According to their results, the system must be at least $2k$-sparse observable in order to tolerate $k$ malicious sensors. This means that the removal of any $2k$ sensors should not compromise the system's observability. Similarly, for the resilient cooperative control problem, a $(2k+1)$-robust graph is necessary to handle $k$ malicious agents, where redundant communication channels should be added to each normal agent \cite{yan2022resilient,leblanc2013resilient}.

In these works, exact system models are integral to the algorithm design. However, deriving models from first principles can be challenging in practical scenarios. Often, only input/output trajectories collected from the system are available, which motivates the study of data-driven control. Typically, data-driven methods can be classified into \textit{indirect} data-driven approaches, involving system identification followed by the model-based control, and \textit{direct} data-driven approaches, where control policies are learned directly from data without identifying a model  \cite{piga2017direct,dorfler2022role,markovsky2022data}.

Here we aim to enhance the CPS resiliency using a direct data-driven approach. We investigate two potential attack scenarios. In the first scenario, up to $k$ entries of the collected data are malicious. On the other hand, the second scenario assumes that at most $k$ channels of sensors or actuators are compromised, meaning that any input/output data collected from these channels can be manipulated. Our objective is to reconstruct the true trajectory for both scenarios without identifying a parametric representation of the system. 

\revise{As shown in \cite{markovsky2021behavioral,markovsky2023data}, the behavioral approach is particularly suitable for solving direct data-driven problems and has advantages over classical model-based methods. The key result established in the behavioral setting that makes solution of the direct-data driven control problems possible is the so called fundamental lemma (\hspace{0.5pt}\cite{willems2005note}). In our paper, we build upon a more recent result from \cite{markovsky2022identifiability}, which extends the fundamental lemma by providing necessary and sufficient conditions—known as the generalized persistency of excitation condition—for establishing a nonparametric representation of the data-generating system (which is the unknown system under study), without requiring a predetermined input/output partitioning of the data \cite{fl}. Specifically, the nonparametric representation refers to an equivalent description of the system using input/output data, rather than the system matrices.} The key assumption on the data requires prior knowledge of the true system's complexity (number of inputs and order). The result is applicable for data consisting of multiple trajectories, uncontrollable systems, and various classes of \revise{nonlinear systems}, see \cite{berberich2020,9304122,9683327,9683151,ddsim-narx}.

\revise{Note that the data used to construct the nonparametric representation is obtained \textit{offline}. Building upon this nonparametric representation, we then perform \textit{online} data reconstruction. That is, to detect whether the data received online is malicious.} The contributions of this paper are \revise{two-fold}:

 1) In both the entry-attacked and channel-attacked scenarios, we formulate the trajectory recovery problem as optimization problems utilizing \cite[Corollary~21]{markovsky2022identifiability}. We introduce \revise{sufficient conditions on the Hankel matrix for performing the online data reconstruction that ensures recovery of the true trajectory}. The conditions are demonstrated to be tight; violating them can result in unsuccessful recovery. 

2) The proposed problems are known to be NP-hard. Finding exact solutions involves solving a combinatorial number of subproblems. To reduce the computational complexity, we approximate the problems by using the $\ell_1$-norm and group Least Absolute Shrinkage and Selection Operator (LASSO). This transformation renders the problems convex, enabling tractable solutions with low computational complexity. We further establish conditions under which these approximation problems also yield the true trajectory. Extension to the noisy data is finally developed.

\revise{To the best of our knowledge, this is the first work to consider security for unknown systems by using the behavioral language. The motivation for ``recovering" the true data is to ensure a safe trajectory required for the subsequent data-driven control. Without this, malicious data could lead to incorrect or even unsafe control actions. Therefore, our data recovery method acts as a safeguard for data-driven approaches, regardless of the specific learning or control algorithms employed.}

The rest of this paper is organized as follows. Section~\ref{sec:pre} introduces preliminaries on behavior language. Sections~\ref{sec:case1} and \ref{sec:case2} formulate the trajectory recovery problem in the entry-attacked and channel-attacked scenarios, respectively. We propose computation-efficient solutions in Section~\ref{sec:solve}, and verify their performance through numerical examples in Section~\ref{sec:sim}. Finally, Section~\ref{sec:cond} concludes the paper.

\section{Preliminaries}\label{sec:pre}

We use the behavioral language \cite{W07,markovsky2022data}. \revise{Let $q$ be the number of variables, $\mathbb{N}$ be the set of natural numbers, and $(\mathbb{R}^q)^\mathbb{N}$ be the set of functions from $\mathbb{N}$ to~$\mathbb{R}^q$. 
A system~$\mathscr{B}$ is defined as a set of trajectories, i.e., $\mathscr{B} \subseteq (\mathbb{R}^q)^\mathbb{N}$.} The system is linear if $\mathscr{B}$ is a subspace, and is time-invariant if it is invariant to the action of the unit shift operator~$\sigma$, defined as $(\sigma w)(t)=w(t+1)$. The class of linear time-invariant (LTI) systems with~$q$ variables is denoted by~$\calL^q$. 

If necessary, the variables $w$ of $\Beh\in\calL^q$ can be partitioned into inputs~$u$ and outputs~$y$, i.e., there is a permutation matrix~$\Pi$, such that $w = \Pi \bsm u\\ y\esm$. An input/output partitioning is in general not unique; however, the number of inputs, denoted by $\invm(\Beh)$, is independent of the partitioning and is therefore a property of the system~$\Beh$. Moreover, we define
$
    \invp(\Beh)\triangleq q-\invm(\Beh),
$ \revise{which is the number of output channels.} Other properties of an LTI system~$\Beh$, used later on, are the order $\invn(\Beh)$ (dimension of a minimal state-space representation) and the lag $\invl(\Beh)$ (degree of a minimal kernel representation). \revise{We refer the readers to the recent overview and tutorial papers \cite{markovsky2021behavioral,markovsky2023data} for more details on the behavioral language.}

We denote by $w|_L$ the restriction of $w\in(\mathbb{R}^q)^\mathbb{N}$ to the interval $[1,L]$, i.e., $\big(w(1),\ldots,w(L)\big)$ and by $\mathscr{B}|_L$ the restriction of $\mathscr{B}$ to the interval $[1, L]$, i.e., $\left.\mathscr{B}\right|_L\triangleq\left\{w|_L \mid w \in\right.$ $\mathscr{B}\}$. 
For an LTI system, $\left.\mathscr{B}\right|_L$ is a subspace of $\mathbb{R}^{q L}$.

For any vector $v$, $v_r$ denotes the $r$-th entry of $v$. Let $\mI = \{i_1,\cdots,i_s\}$ be a set of indices with $|\mI|=s$. For $v\in\mathbb{R}^m$, we define $v|_{\mI}$ as the subvector of $v$ with indices in $\mI$, i.e.,
\begin{equation*}
	v|_{\mI} \triangleq [v_{i_1}, \ldots, v_{{i}_s}]^\top \in \mathbb{R}^s.
\end{equation*}
Similarly, for a matrix $A\in\mathbb{R}^{m\times n}$, $A|_{\mI}$ is defined as the submatrix of $A$ with row indices in ${\mI}$:
\begin{equation*}
	A|_{\mI} \triangleq [A^\top_{{i}_1}, \ldots, A^\top_{{i}_s}]^\top \in \mathbb{R}^{s\times n},
\end{equation*}
where $A_{{i}_r}$ represents the ${i}_r$-th row of $A$.
Furthermore, given a binary vector (i.e., a vector which entries are either $0$ or $1$), $\idx(v)$ returns an index set containing the positions of $1$-entries. 


\section{\revise{Entry-Attacked Scenario}}\label{sec:case1}
We consider a discrete-time LTI system $\Beh\in\calL^q$.
Our focus lies on a scenario where the system is unknown, but input/output trajectories from the system are available, enabling the data-driven approach. This section assumes that the collected data is noiseless. Later in Section~\ref{sec:noise}, we show how to extend the results to noisy datasets.

Consider a $T$-length trajectory $w_d\in(\R^{q})^T$ of the system, i.e., $w_{d} \in \mathscr{B}|_T$.
Let $\mathscr{H}_L(w_d)$ denote the Hankel matrix formed by $w_d$ with $L$ block rows, where $1\leq L \leq T$: 
\begin{equation*}
\label{eqn:hankel}	
\begin{split}
  \mathscr{H}_L(w_{d}) &\triangleq \begin{bmatrix}
     w_{d}(1) & w_{d}(2) & \cdots & w_{d}(T-L+1) \\
		w_{d}(2) & w_{d}(3) & \cdots & w_{d}(T-L+2) \\
		\vdots & \vdots & & \vdots \\
		w_{d}(L) & w_{d}(L+1) & \cdots & w_{d}(T)
 \end{bmatrix} \\ &\in \mathbb{R}^{q L \times (T-L+1)}.      
\end{split}
\end{equation*}   
We assume that $w_d$ is attack-free. Otherwise, it is not possible to recover the trajectory since $w_d$ encodes an understanding of the underlying system.

In this paper, we will use the following result as a key solution technique \cite[Corollary~21]{markovsky2022identifiability}. 
\begin{lemma}\label{lmm:funda_lemma}
	Consider $w_{d} \in \mathscr{B}|_T$, where $\mathscr{B} \in \mathscr{L}^q$ and $L \geq \invl(\Beh)$. The following statements are equivalent:
	\begin{enumerate}
		\item Generalized persistent excitation condition holds\footnote{\revise{Prior knowledge of the true system complexity, such as the number of inputs and order, is needed for the generalized persistent excitation condition. However, this is not required in the methods proposed later in the paper.}}, namely, \begin{equation}\label{eqn:PE}
			\rank \mathscr{H}_L\left(w_{d}\right)=\invm(\Beh) L+\invn(\Beh).
		\end{equation}
		\item Any $\Bar{w}\in(\R^{q})^L$ is a trajectory of $\mathscr{B}$, that is, $\Bar{w}\in \mathscr{B}|_L$, if and only if there exists a $g\in\mathbb{R}^{T-L+1}$ such that
		\begin{equation}\label{eqn:map}
			\Bar{w} = \mathscr{H}_L\left(w_{d}\right) g.
		\end{equation} 
  That is, $\mathscr{B}|_L=\text {image } \mathscr{H}_L\left(w_{d}\right).$ 
	\end{enumerate}
\end{lemma}

Lemma~\ref{lmm:funda_lemma} is related to the fundamental lemma~\cite{willems2005note}. However, it does not require an input/output partitioning of the variables nor controllability of the system~$\mathscr{B}$. Instead of a condition on an input component of the trajectory (persistency of excitation of an input), it imposes a condition on the data $w_d$ (generalized persistency of excitation). Finally, the conditions of Lemma~\ref{lmm:funda_lemma} are necessary and sufficient while the ones of \cite{willems2005note} are only sufficient \cite{fl}.

\subsection{Entry-attacked model}

In the context of data-driven control, trajectories collected from the system directly influence controller design. Hence, ensuring reliability of the data is paramount. This paper thus focuses on the problem of secure trajectory reconstruction under data manipulation.

Let $\bar{w} \in(\R^{q})^L$ denote an $L$-length \revise{trajectory} of the LTI system $\Beh\in\calL^q$. When subjected to attacks, some entries of $\bar{w}$ are tampered with. As a result, a manipulated data sequence $w \in(\R^{q})^L$ is received. Our aim is to reconstruct the true signal $\bar{w}$ from the received data $w$. 

In the entry-attacked model, we assume that the attacker can change these entries to
arbitrary values, but can only compromise at most $k$ entries. This is reasonable since the adversary usually has limited energy. Let $\mathcal{C}\subseteq \{1,\ldots,qL\}$ denote the index set of malicious entries. 
Note that while the number $k$ is assumed to be known \textit{a priori}, the exact set $\mathcal{C}$ is unknown.

\subsection{A brute force algorithm}
To recover the real data, our goal is to identify a ``legal" trajectory of the system that matches the received data except for $k$ entries. This leads to the following optimization problem:
\begin{equation}\label{eqn:P0_}
\begin{split}
	\min_{\tilde{w},\mathcal{I}} \qquad&||w|_{\mI}-\tilde{w}|_{\mI}||_2\\
	s.t. \qquad &|\mI| = qL-k,\\
& \tilde{w} \in \mB|_L.
\end{split}
\end{equation}
Using Lemma~\ref{lmm:funda_lemma}, the above problem becomes
\begin{equation}\label{eqn:P0}
\begin{split}
	\min_{\tilde{w},g,\mathcal{I}} \qquad&||w|_{\mI}-\tilde{w}|_{\mI}||_2\\
	s.t. \qquad &|\mI| = qL-k,\\
& \tilde{w} = \mathscr{H}_L\left(w_{d}\right) g.
\end{split}
\end{equation}

Problem \eqref{eqn:P0} is generally NP-hard, as it can be re-formulated as a mixed integer linear program (MILP). In Algorithm~\ref{alg:combinatorial}, we present a brute force algorithm for solving the problem, where a combinatorial number of subproblems need to be solved to identify the index set $\mI^{(i)}$ leading to the minimum. Specifically, there are $\binom{qL}{k}$ possible sets of benign entries. That is, $|\mathscr{I}|=\binom{qL}{k}$, where
$
\revise{\mathscr{I} \triangleq\{\mathcal{I}^{(i)}| \mathcal{I}^{(i)} \subseteq [qL],|\mathcal{I}^{(i)}|=qL-k\}.}
$
\revise{We consider all possible subproblems by iterating through every set $\mathcal{I}^{(i)}\in \mathscr{I}$.}
For each subproblem $i$, the algorithm aims to solve the equation\footnote{The notation $w|_{\mI}$
for $\mI$ being a set overloads the notation $w|_L$ defined earlier in Section~\ref{sec:pre}
for $L$ being an integer.}
\begin{equation}\label{eqn:g}
w|_{\mI^{(i)}} = \mathscr{H}_L(w_d)|_{\mI^{(i)}} g.
\end{equation}
If a solution $g$ exists, Algorithm~\ref{alg:combinatorial} terminates with the output
\begin{equation}\label{eqn:tilde_w}
\tilde{w} = \mathscr{H}_L\left(w_d\right) g.
\end{equation}
By doing so, the algorithm identifies a trajectory that fits $w$ while excluding $k$ entries. Alternatively, if no solution is found for any $i=1,\ldots,\binom{qL}{k}$, Algorithm~\ref{alg:combinatorial} finishes without producing a solution. 

The complexity of Algorithm~\ref{alg:combinatorial} increases exponentially with the size of data and the number of malicious entries. In Section~\ref{sec:solve}, we will introduce a computationally efficient approximation for solving \eqref{eqn:P0}. 

\begin{algorithm}
\caption{Combinatorial algorithm for solving problem~\eqref{eqn:P0}.}\label{alg:combinatorial}
\begin{algorithmic}
\REQUIRE $w_d$, $w$, and $k$
\FOR{$i \in \{1,\ldots, \binom{qL}{k}\}$}
 \STATE Select $\mI^{(i)}$ with $|\mI^{(i)}| = qL-k$.
\IF{$\rank \left[\mathscr{H}_L(w_{d})|_{\mI^{(i)}} \;\; w|_{\mI^{(i)}}\right]=\rank  \mathscr{H}_L(w_{d})|_{\mI^{(i)}}$}
\STATE Find $g^{(i)}$ by solving \eqref{eqn:g} and compute $\tilde{w}$ by using \re{eqn:tilde_w}.
\STATE \algorithmicbreak
\ENDIF
\ENDFOR
\end{algorithmic}
\end{algorithm}

{\color{blue}
\begin{remark}\label{rmk:k}
As defined, \( k \) is an upper bound on the number of malicious entries. In practice, we usually can only estimate this upper bound instead of knowing the exact number of malicious entries.  Therefore, we design the secure algorithms based on the worst-case scenario, which involves removing exactly \( k \) entries in Algorithm~\ref{alg:combinatorial}. Since the algorithm is designed to handle this worst case, it is also effective when fewer than \( k \) entries are malicious.

It is true that when fewer than \( k \) entries are adversarial, we may end up discarding more data than necessary. To mitigate this, we can start with a smaller value \( \bar{k} < k \). If no solution is found for a given \( \bar{k} \), it indicates that \( \bar{k} \) is too small, that is, we are not removing enough entries. In this case, we iteratively increase \( \bar{k} \) until Algorithm~\ref{alg:combinatorial} finds a solution.
\end{remark}}

\subsection{Performance analysis}\la{sec:method}
\revise{Algorithm~\ref{alg:combinatorial} searches through all possible sets $\mathcal{I}^{(i)}$ such that $|\mathcal{I}^{(i)}| = qL - k$. Among these, there must exist a set $\mathcal{I}^{(i)}$ that contains only the indices of benign entries, that is, $\mI^{(i)} \subset [qL]/\mathcal{C}$, as the total number of benign entries is no less than $qL - k$. 
In this case, all compromised entries are removed. Therefore, we can always solve \eqref{eqn:g}} and Algorithm~\ref{alg:combinatorial} produces the true trajectory $\tilde{w} = \bar{w}$ \cite[Proposition~9]{markovsky2022data}. However, a natural question arises: when \eqref{eqn:g} yields a solution, does this necessarily imply that we have excluded all malicious entries and $\tilde{w} = \bar{w}$? More generally, what conditions guarantee that the output of Algorithm \ref{alg:combinatorial} always coincides with the real trajectory? This subsection is dedicated to addressing these questions.

To this end, let us first define the \textit{minimum critical row set} for a matrix.
\begin{definition}[Minimum critical row set]\label{def:mini_set}
	For a matrix $A\in\mathbb{R}^{m\times n}$, an index set $\mathcal{S}(A)$ is a \textit{critical row set} for $A$, if removing all the rows specified in the set $\mathcal{S}(A)$ causes the rank of~$A$ to decrease by~$1$. That is:
	\begin{equation*}
		\rank A|_{[m]/\mathcal{S}(A)} = \rank A -1.
	\end{equation*}
	The \textit{minimum critical row set}, $\mathcal{S}^*(A)$, is a critical row set with the least cardinality and is denoted as $\mathcal{S}^*(A)$.
\end{definition}

The minimum critical row set identifies the least number of rows, the removal of which causes the matrix $A$ to lose rank. The minimum critical row set for the Hankel matrix $\mathscr{H}_L(w_d)$ plays a crucial role in ensuring the successful reconstruction of system trajectory, as it allows $\mathscr{H}_L(w_d)$ to ``reproduce" the true trajectory even with the removal of certain rows. To establish this, we require the following condition.

\begin{condition}\label{con:mini_set}
	The cardinality of the minimum critical row set for the Hankel matrix $\mathscr{H}_L(w_d)$ is at least $2k+1$, i.e.,  $|\mathcal{S}^*(\mathscr{H}_L(w_d))| \geq 2k+1$.
\end{condition}


When Condition~\ref{con:mini_set} holds, removing any $2k$ rows does not lead to $\mathscr{H}_L(w_d)$ losing rank. In general, verifying this condition is computationally difficult. 
A related concept in model-based control is the $2k$-sparse observability \cite{mao2019secure, nakahira2018attack,shoukry2017secure,li2021low}. It requires that, removing any $2k$ sensors should not compromise observability of the system, that is, the observability matrix retains full rank. 
Existing works \cite{mao2022computational} and \cite{li2021low} offer potential insights into assessing the $2k$-sparse observability, which, therefore, shed light on checking Condition~\ref{con:mini_set}; We omit the detailed discussion here. Instead, the following lemma establishes an upper bound of $\big|\mathcal{S}^*(\HH_L(w_d))\big|$.

\begin{lemma}\label{lmm:S}
Under the generalized persistency of excitation condition \re{eqn:PE}, 
$\big|\mathcal{S}^*(\HH_L(w_d))\big| \leq \invp(\Beh)+1$.
\end{lemma}
\begin{proof}
Let us define \[\mathcal{I}_1 \triangleq \{\,1,\ldots,q(L-1)\,\}\] and consider 
$ \HH_L(w_d)|_{\mathcal{I}_1}$, which is the submatrix of $\HH_L(w_d)$ by removing the last $q$ rows.  Clearly, $ \HH_L(w_d)|_{\mathcal{I}_1}$ is also a Hankel matrix, which has $(L-1)$ block rows and is formed by the first $(T-1)$-length input/output trajectories of $w_d$. That is, $\HH_L(w_d)|_{\mathcal{I}_1}=\HH_{L-1}(w_d|_{[q(T-1)]})$.

Since the condition \re{eqn:PE} holds, by Lemma~\ref{lmm:funda_lemma}, for any $\bar{w}\in \mathscr{B}|_L$, there exists a $g\in\mathbb{R}^{T-L+1}$ such that
$\bar{w} = \mathscr{H}_L\left(w_{d}\right) g,
$ implying that $v = \HH_L(w_d)|_{\mathcal{I}_1} g$ holds for any $v\triangleq \bar{w}|_{\mathcal{I}_1}\in \mathscr{B}|_{L-1}$. Recalling Lemma~\ref{lmm:funda_lemma} again, we know that $\rank\HH_L(w_d)|_{\mathcal{I}_1} = m(L-1) + n$.

Next, consider 
\[\mathcal{I}_2 \triangleq \{\,1,\ldots,qL-\invp(\Beh)-1\,\}\] 
and $\HH_L(w_d)|_{\mathcal{I}_2}$, i.e., $\HH_L(w_d)$ with its last $\invp(\Beh)+1$ rows removed. We have that 
\begin{equation}
\begin{split}
\rank\HH_L(w_d)|_{\mathcal{I}_2} &\leq \rank\HH_L(w_d)|_{\mathcal{I}_1} + m - 1 \\&= \rank\HH_L(w_d)-1.    
\end{split}    
\end{equation}
Therefore, $[qL] / \mathcal{I}_2$ is a critical row set for $\HH_L(w_d)$.
\end{proof}

With this condition, we derive the following results, which ensures that Algorithm~\ref{alg:combinatorial} recovers the correct trajectory. For notation simplicity, henceforth, we denote 
\[
	\mathscr{H} \triangleq \mathscr{H}_L(w_{d}).
\]

\begin{theorem}\label{lmm:sufficient1}
    Suppose that \revise{no more than $k$ entries are malicious}, \eqref{eqn:PE} is satisfied, and Condition~\ref{con:mini_set} holds. When Algorithm~\ref{alg:combinatorial} finds a solution $\tilde{w}$, it follows that $\tilde{w}=\bar{w}$. That is, the real trajectory is recovered.
\end{theorem}
\begin{proof}
Since Algorithm~\ref{alg:combinatorial} outputs a solution, there exists a set \revise{$\mI^{(i^*)}$} such that \eqref{eqn:g} is solvable.
We denote 
\revise{\begin{equation}\label{eqn:B}
    \mathcal{C}^{(i^*)} \triangleq \mathcal{I}^{(i^*)} \cap \mathcal{C}, \ \mathcal{B}^{(i^*)} \triangleq \mathcal{I}^{(i^*)} \backslash \mathcal{C}^{(i^*)},
\end{equation}}
which represent subsets of \revise{$\mI^{(i^*)}$} containing all malicious and benign entries, respectively. 

Since $|\mI^{(i^*)}| = qL-k$ and there are \revise{at most $k$ malicious entries}, it is clear that $|\mathcal{B}^{(i^*)}|\geq qL-2k$. As Condition~\ref{con:mini_set} holds, we conclude 
\begin{equation}\label{eqn:rank}
	\rank \mathscr{H}|_{\mathcal{B}^{(i^*)}} = \rank \mathscr{H}|_{\mathcal{I}^{(i^*)}} = \rank \mH.
\end{equation}  
\revise{Since $\mathscr{H}|_{\mathcal{B}^{(i^*)}}$ and $\mathscr{H}|_{\mathcal{I}^{(i^*)}}$ are obtained from $\mH$ by removing rows, it follows that}
\begin{equation}\label{eqn:null}
	\nul \mathscr{H}|_{\mathcal{B}^{(i^*)}} = \nul \mathscr{H}|_{\mathcal{I}^{(i^*)}} = \nul \mathscr{H}.
\end{equation}
Since \eqref{eqn:g} has a particular solution, denoted by $g_0^{(i^*)}$, we can represent the solution set of \eqref{eqn:g} as $\mathcal{G}^{(i^*)} = \{ g_0^{(i^*)}+\tilde{g}^{(i^*)}| \tilde{g}^{(i^*)}\in \nul \mathscr{H}|_{\mathcal{I}^{(i^*)}}\}$. On the other hand, we denote $\mathcal{G}_{\mathcal{B}}^{(i^*)}$ as the solution set of the following equation
\begin{equation}\label{eqn:solveB}
w|_{\mathcal{B}^{(i^*)}}=\mathscr{H}|_{\mathcal{B}^{(i^*)}} g.
\end{equation}
Notice that $g_0^{(i^*)}$ also solves \eqref{eqn:solveB}. Hence, \eqref{eqn:null} leads to
$
	\mathcal{G}^{(i^*)} = \mathcal{G}_{\mathcal{B}}^{(i^*)}.
$
That is, any solution of \eqref{eqn:solveB} also solves \eqref{eqn:g} and verse vice. In view of \eqref{eqn:tilde_w} and \eqref{eqn:null}, any solution $g\in \mathcal{G}^{(i^*)}$ (or $\mathcal{G}_{\mathcal{B}}^{(i^*)}$) maps to the same $\tilde w$. 

Next, we show $\tilde{w}=\bar{w}$. It is sufficient to prove that, there exists $g \in \mathcal{G}_{\mathcal{B}}^{(i^*)}$ such that 
$\bar{w} =\mH g$. To see this, we notice that $\bar{w}$ is the true trajectory. Therefore, by the virtue of Lemma~\ref{lmm:funda_lemma}, there must exist $
\bar{g}\in\mathbb{R}^{T-L+1}$ such that 
\begin{equation}\label{eqn:gbar}
    \bar{w} = \mH \bar{g}
\end{equation}
As $\mathcal{B}^{(i^*)}$ only contains non-faulty entries, it is clear
$
	w|_{\mathcal{B}^{(i^*)}} = \mathscr{H}|_{\mathcal{B}^{(i^*)}} \bar{g}.
$
Therefore, $\bar{g}\in \mathcal{G}_{\mathcal{B}}^{(i^*)}$, completing the proof.
\end{proof}

Condition~\ref{con:mini_set} is a sufficient condition to guarantee the trajectory recovery by Algorithm~\ref{alg:combinatorial}. \revise{Moreover, when Condition~\ref{con:mini_set} holds, Algorithm~\ref{alg:combinatorial} always produces a unique and optimal solution. Therefore, the problem is well-defined.}
In the following proposition, we further  demonstrate that this condition is tight.

\begin{proposition}\label{lmm:necessary}
Suppose that \revise{no more than $k$ entries are malicious}, \eqref{eqn:PE} is satisfied, but Condition~\ref{con:mini_set} does not hold. Even if Algorithm~\ref{alg:combinatorial} yields a solution $\tilde{w}$, it is possible that $\tilde{w}\neq \bar{w}$.
\end{proposition}
\begin{proof}
We will prove this by construction. Since $|\mathcal{S}^*(\mathscr{H})| \leq 2k$, an adversary selects the set of manipulated entries as 
 \begin{equation}\label{eqn:C}
     \begin{cases}
       \mathcal{C} =  \mathcal{S}^*(\mathscr{H}), \text{ if } |\mathcal{S}^*(\mathscr{H})| < k,\\
       \mathcal{C} \subset  \mathcal{S}^*(\mathscr{H}) \text{ with } |\mathcal{C}|=k, \text{ otherwise}. 
     \end{cases}
 \end{equation}
Let us consider any $\mI^{(i)}$ such that $\mI^{(i)}\supseteq \mathcal{C}$. Then
$\mathscr{H}|_{\mathcal{C}^{(i)}}=\mathscr{H}|_{\mathcal{C}},$
where $\mathcal{C}^{(i)}$ is defined in \eqref{eqn:B}.
 With respect to the two cases in \eqref{eqn:C}, it follows that either $\mathcal{B}^{(i)}\subseteq [qL]\backslash \mathcal{S}^*(\mathscr{H})$, or $|\mathcal{B}^{(i)}|=qL-2k$. Therefore, there must exist some set $\mI^{(i)}$ such that
$\rank \mathscr{H}|_{\mathcal{B}^{(i)}} < \rank \mathscr{H}|_{\mathcal{I}^{(i)}}.$
 As a result, it is not difficult to verify 
$\nul \mathscr{H}|_{\mathcal{B}^{(i)}} \supsetneq \nul \mathscr{H}|_{\mathcal{I}^{(i)}}. $
 
Hence, the adversary can always choose some $\tilde{g} \in \{\nul \mathscr{H}|_{\mathcal{B}^{(i)}}\}\backslash \{\nul \mathscr{H}|_{\mathcal{I}^{(i)}}\}$ such that
\[
\begin{split}
     \mathscr{H}|_{\mathcal{B}^{(i)}} \tilde{g} &= 0,\;\mathscr{H}|_{\mathcal{C}} \tilde{g} \neq 0.
\end{split}
\]
Let $
    g = \bar{g} + \tilde{g},$ where $\bar{g}$ generates the real trajectory as given in \eqref{eqn:gbar}. We have
    \begin{equation}\label{eqn:wb}
        \mathscr{H}|_{\mathcal{B}^{(i)}} g = \mathscr{H}|_{\mathcal{B}^{(i)}} (\bar{g} + \tilde{g})= w|_{\mathcal{B}^{(i)}}.
    \end{equation}The adversary thus modifies data at the manipulated entries as 
 \begin{equation*}
   w|_{\mathcal{C}}=\mathscr{H}|_{\mathcal{C}}g = \mathscr{H}|_{\mathcal{C}}(\bar{g} + \tilde{g})=\Bar{w}|_{\mathcal{C}}+\mathscr{H}|_{\mathcal{C}}\tilde{g} \neq \Bar{w}|_{\mathcal{C}}.  
 \end{equation*}
Namely, it changes $\Bar{w}|_{\mathcal{C}}$ to a different signal $w|_{\mathcal{C}}$.
We show that by doing so, it can deceive the system operator into believing in some false data $\tilde{w}$ with $\tilde{w}\neq \bar{w}$. To see this, without loss of generality, we assume
\begin{equation}
    \mathcal{I}^{(i)} = \begin{bmatrix}
        \mathcal{B}^{(i)}\\
        \mathcal{C}
    \end{bmatrix}.
\end{equation}
Therefore,
\[
    w|_{\mathcal{I}^{(i)}} = \begin{bmatrix}
        w|_{\mathcal{B}^{(i)}}\\
        w|_{\mathcal{C}}
    \end{bmatrix} = \begin{bmatrix}
        \mathscr{H}|_{\mathcal{B}^{(i)}}\\
        \mathscr{H}|_{\mathcal{C}}
    \end{bmatrix}g = \mathscr{H}|_{\mathcal{I}^{(i)}} g
\]
That means $g$ is a solution of \eqref{eqn:g}. As a result, Algorithm~\ref{alg:combinatorial} can solve \eqref{eqn:g} and output $\tilde{w}$. However, it is clear
\[
   \tilde{w} = \mH g  = \begin{bmatrix}
        \mathscr{H}|_{\mathcal{C}} \\
        \mH|_{[qL]\backslash \mathcal{C}}
    \end{bmatrix}g = \begin{bmatrix}
        w|_{\mathcal{C}}\\
        \bar{w}|_{[qL]\backslash \mathcal{C}}
    \end{bmatrix} \neq \bar{w}.
\]
\end{proof}

\revise{Combining the above theorems with Lemma~\ref{lmm:S}, for successful data reconstruction, it must hold that
	$
	2k + 1 \leq \invp(\Beh) + 1,
	$
	namely,
	$
	k \leq \frac{\invp(\Beh)}{2}.
	$
	In other words, the number of attacked entries must \textit{not} exceed $\frac{\invp(\Beh)}{2}$, where $\invp(\Beh)$ is defined as the number of output channels. Otherwise, the data recovery problem can be infeasible.}

\revise{\begin{remark}\label{rmk:condition}
	There may be cases where Condition~\ref{con:mini_set} does not hold, yet Algorithm~\ref{alg:combinatorial} still successfully reconstructs the true trajectory, particularly when the adversary is less aggressive. However, Condition~\ref{con:mini_set} is tight, as shown in Proposition~\ref{lmm:necessary}, where violating this condition leads to failure in recovering the true trajectory. Since this paper focuses on the worst-case scenario, Condition~\ref{con:mini_set} is essential.
\end{remark}}

\section{Channel-Attacked Scenario}\label{sec:case2}
In Section~\ref{sec:case1}, we investigate the scenario where up to $k$ entries in the collected trajectory are manipulated. However, in real-world scenarios, adversaries may exploit specific input or output \textit{channels}, consistently altering all data coming from these channels. This section is devoted to investigating this scenario, where a maximum of $k$ out of $q$ channels are susceptible to manipulation. Since the adversary possesses complete control over these channels, any data collected from them can be falsified. In this case, up to $qk$ entries in the received data $w\in(\R^{q})^L$ might be different from those in the real trajectory $\bar w$. That is, $||w-\bar w||_0\leq qk.$

Let us denote $\mathcal{C} \subseteq \{1,\ldots,q\}$ as the index set of the manipulated channels, where $|\mathcal{C}|\leq k$. Similar to the previous case, the value of $k$ is assumed to be known \textit{a priori}, but the specific elements in the set $\mathcal{C}$ remain unknown. Notice that here, we assume that the set of compromised channels remains constant over time, as widely adopted in the literature (e.g., \cite{ren2018binary,fawzi2014secure,yan2022resilient}).

\subsection{A brute force algorithm}
To reconstruct the true trajectory from $w$, we solve the following problem:
\begin{equation}\label{eqn:P1}
\begin{split}
\min_{g,v_o} \qquad&||v(w-\tilde{w})||_2\\
s.t. \qquad  
& v_o \in \{0,1\}^q,\\
& ||v_o||_0 = q-k,\\
&v = \1_L \otimes v_o,\\
&\tilde{w} = \mathscr{H}_L\left(w_{d}\right) g.
\end{split}
\end{equation}
Here, the binary vector $v_o$ has $(q-k)$ entries set to $1$ and the other $k$ entries set to $0$. 
Moreover, by Lemma~\ref{lmm:funda_lemma}, $\tilde{w}$ is a valid trajectory of the system $\mathscr{B}$. Therefore, by optimizing over $g$ and $v_o$, our objective is to find the most fitting trajectory among all valid ones, such that it matches the received data $w$ at $(q-k)$ channels corresponding to the $1$-entries of $v_o$. \revise{Note that \eqref{eqn:P1} deals with the worst-case situation where exactly $k$ channels are removed. Therefore, it also works safely when fewer than $k$ channels are malicious.}

To obtain an exact solution to Problem~\eqref{eqn:P1}, a combinatorial problem needs to be tackled as outlined in Algorithm~\ref{alg:case2}. Specifically, within each of the $\binom{q}{k}$ subproblems, indexed as $i$, we pre-select the vector $v_o^{(i)} \in \{0,1\}^q$ with $||v_o^{(i)}||_0 = q-k$ and compute the corresponding index set $\mI^{(i)}$, specifying the positions of entries to be included in solving \eqref{eqn:g2}. If a solution is not found for any $i=1,\ldots,\binom{q}{k}$, Algorithm~\ref{alg:case2} terminates without producing a solution. 

\begin{algorithm}
	\caption{Combinatorial algorithm for solving problem~\eqref{eqn:P1}.}\label{alg:case2}
	\begin{algorithmic}
  \REQUIRE $w_d$, $w$, and $k$
\FOR{$i \in \{1,\ldots, \binom{q}{k}\}$}
 \STATE Select $v_o^{(i)} \in \{0,1\}^q$ with $||v_o^{(i)}||_0=q-k$.
 \STATE Let $v^{(i)} = \1_L \otimes v_o^{(i)}$ and $\mI^{(i)} = \idx(v^{(i)})$. 
\IF{$\rank \left[\mathscr{H}|_{\mI^{(i)}} \;\; w|_{\mI^{(i)}}\right]=\rank  \mathscr{H}|_{\mI^{(i)}}$}
\STATE Get $g^{(i)}$ by solving
\begin{equation}\label{eqn:g2}
\left.w\right|_{\mI^{(i)}}=\left.\mathscr{H}\right|_{\mI^{(i)}} g^{(i)}.
\end{equation}
\STATE Calculate 
\begin{equation}\label{eqn:tilde_w2}
\begin{split}
	\tilde{w} &=\mathscr{H} g^{(i)}.
\end{split}
\end{equation}
\STATE \algorithmicbreak
\ENDIF
\ENDFOR
\end{algorithmic}
\end{algorithm}

\subsection{Performance analysis}
In contrast to the entry-attacked scenario as discussed in Section~\ref{sec:case1}, in Algorithm~\ref{alg:case2}, the rows to be removed in $\mathscr{H}$ are determined by channels rather than individual entries. Moreover, to reconstruct the original trajectory, it is crucial to ensure that the matrix formed by the remaining rows has sufficient rank to represent the full space of the system $\mathscr{B}$. To address this, we introduce the following definitions, quantifying the rank of a matrix after deleting rows corresponding to specific channels:

\begin{definition}[$(q, L)$-periodical set]
Given an index set $\mathcal{S}\subseteq \{1, \ldots, q\}$, we define $\mathscr{C}^{q, L}(\mathcal{S})\subseteq \{1,\ldots,qL\}$ as the $(q, L)$-periodical set of $\mathcal{C}$, such that
\[
\mathscr{C}^{q, L}(\mathcal{S}) \triangleq \bigcup_{ i\in \mathcal{S}} \mathscr{C}_i^{q, L},
\]
where 
\begin{equation}\label{eqn:Ci}
\mathscr{C}_i^{q, L} \triangleq \{i+ q\ell |\ell \in \mathbb{N}, 0\leq \ell \leq L-1\}.
\end{equation}
\end{definition}

Clearly, for any $j$ in the $(q, L)$-periodical set $\mathscr{C}^{q, L}(\mathcal{S})$, there exists $i \in \mathcal{S}$ such that $j \mod q = i$. Furthermore, we introduce a \textit{$(q, L)$-critical row set} for any matrix $A$, if removing the rows specified by this set results in a loss of rank for $A$.

\begin{definition}[Minimum $(q, L)$-critical row set]\label{def:min_period_set}
	Consider any matrix $A\in\mathbb{R}^{m\times n}$. Let us define $\widetilde{\mathcal{S}}(A)$ as a \textit{$(q, L)$-critical row set} for $A$, if removing all the rows specified in its $(q, L)$-periodical set $\mathscr{C}^{q, L}(\widetilde{\mathcal{S}}(A))$ causes a rank loss for $A$. That is:
	\[
	\rank A|_{[mL]/\mathscr{C}^{q, L}(\widetilde{\mathcal{S}}(A))} = \rank A -1.
	\]
	Furthermore, the \textit{minimum $(q, L)$-critical row set} is defined as a $(q, L)$-critical row set with the least cardinality, and is denoted as $\widetilde{\mathcal{S}}^*(A)$.
\end{definition}

\revise{
\begin{remark}
From Definitions~\ref{def:mini_set}--\ref{def:min_period_set}, for any matrix $A$, if \(|\widetilde{\mathcal{S}}^*(A)| = \beta\), there exist \(\beta L\) rows such that removing them causes 
$A$ to lose rank. Therefore, \(|\mathcal{S}^*(A)| \leq \beta L\).
\end{remark}}

Next, we introduce the following condition on the Hankel matrix $\mathscr{H}_L(w_d)$.

\begin{condition}\label{con:cond2}
The cardinality of the minimum $(q, L)$-critical row set for the Hankel matrix $\mathscr{H}_L(w_d)$ is at least $2k+1$, i.e., $|\widetilde{\mathcal{S}}^*(\mathscr{H}_L(w_d))| \geq 2k+1$.
\end{condition}


Recall that the total number of channels is $q$. For any channel $i \in \{1, \ldots, q\}$, rows corresponding to the data from this channel are indexed as $\{i, i+q, \ldots, i+(L-1)q\}$, consistent with the definition of $\mathcal{C}_i^{q, L}$ in \eqref{eqn:Ci}. Therefore, Condition~\ref{con:cond2} requires that the removal of all rows corresponding to the data from \textit{any} $2k$ channels will not result in rank loss for $\mathscr{H}_L(w_d)$. This key aspect ensures the resilience of Algorithm~\ref{alg:case2} with respect to an unknown set of manipulated channels, as proved below:

\begin{theorem}
	Suppose that \revise{no more than $k$ channels are malicious}, \eqref{eqn:PE} is satisfied, and Condition~\ref{con:cond2} holds. When Algorithm~\ref{alg:case2} finds a solution $\tilde{w}$, it follows that $\tilde{w}=\bar{w}$. That is, the real trajectory is recovered.
\end{theorem}
\begin{proof}
Let us denote \revise{$\mI_o^{(i^*)} \triangleq \idx(v_o^{(i^*)})$}, namely, the index set of the selected channels in the $i$-th \revise{subproblem} that solves \eqref{eqn:g2}. It is clear that $|\mI_o^{(i^*)}|=q-k$. We further define
   \revise{ \[
    \mathcal{C}_o^{(i^*)} \triangleq \mathcal{I}_o^{(i^*)} \cap \mathcal{C}, \; \mathcal{B}_o^{(i^*)} \triangleq \mathcal{I}_o^{(i^*)} \backslash \mathcal{C}_o^{(i^*)}.
\]}
As assumed, \revise{$|\mathcal{C}|\leq k$}. \revise{We therefore conclude that $|\mathcal{B}_o^{(i^*)}|=|\mI_o^{(i^*)}|-|\mathcal{C}_o^{(i^*)}|\geq q-2k$.} That is, for any subsystem, the number of benign selected channels is no less than $q-2k$. Let 
$
\mathcal{B}^{(i^*)} \triangleq \mathscr{C}^{q,L}(\mathcal{B}_o^{(i^*)}).
$
Under Condition~\ref{con:cond2}, 
we conclude 
\begin{equation}\label{eqn:rank}
	\rank \mathscr{H}|_{\mathcal{B}^{(i^*)}} = \rank \mathscr{H}|_{\mathcal{I}^{(i^*)}} = \rank \mH.
\end{equation}  
The remainder of this proof follows a similar structure to that of Theorem~\ref{lmm:sufficient1}.
\end{proof}

Finally, we can prove that Condition~\ref{con:cond2} is tight.
\begin{proposition}
Suppose that \revise{no more than $k$ channels are malicious}, \eqref{eqn:PE} is satisfied, but Condition~\ref{con:cond2} does not hold. Even if Algorithm~\ref{alg:case2} yields a solution $\tilde{w}$, it is possible that $\tilde{w}\neq \bar{w}$.
\end{proposition}
\begin{proof}
   The proof of this result follows similar reasoning to that in Proposition~\ref{lmm:necessary}. To be specific, an adversary can set the manipulated channels as
  \begin{equation}\label{eqn:C}
     \begin{cases}
       \mathcal{C} =  \widetilde{\mathcal{S}}^*(\mathscr{H}), \text{ if } |\widetilde{\mathcal{S}}^*(\mathscr{H})| < k,\\
       \mathcal{C} \subset  \widetilde{\mathcal{S}}^*(\mathscr{H}) \text{ with } |\mathcal{C}|=k, \text{ otherwise}. 
     \end{cases}
 \end{equation}
Hence, for any $\mI_o^{(i)}$ such that $\mI_o^{(i)}\supseteq \mathcal{C}$, it holds 
$
    \mathcal{C}_o^{(i)} = \mathcal{C},
$
and
$\mathscr{H}|_{\mathcal{C}^{(i)}}=\mathscr{H}|_{\mathcal{C}},$ where $\mathcal{C}^{(i)} \triangleq \mathscr{C}^{q,L}(\mathcal{C}_o^{(i)})$. Following similar arguments as in Proposition~\ref{lmm:necessary}, we reach the conclusion.
\end{proof}

\revise{Based on the above results, Condition~\ref{con:cond2} is crucial to well-define the problem. That is, in order to ensure the recovery of the true trajectory, removing the rows corresponding to any $2k$ channels must not compromise the rank of $\mH$. Otherwise, the problem becomes infeasible. Moreover, for the same reason discussed in Remark~\ref{rmk:condition}, this condition is essential in working against the worst-case attacks.}

\section{\revise{Computationally-Efficient Solutions}}\label{sec:solve}
The previous sections presented brute force algorithms for solving \eqref{eqn:P0} and \eqref{eqn:P1}. Although these algorithms successfully reconstruct the true trajectory, their complexity grows exponentially with the size of data and the number of malicious entries or channels. In this section, we introduce convex approximations of \eqref{eqn:P0} and \eqref{eqn:P1}. Moreover, we show conditions under which these approximations also yield the real trajectory.

\subsection{Approximation solutions}
\revise{Since a non-convex cardinality constraint is involved in \eqref{eqn:P0}, a natural way to relax it is using an $\ell_1$-norm}:
\begin{equation}\label{eqn:lasso}
	\begin{split}
		g^* = \argmin_{g} \quad& ||w-\mathscr{H}_L\left(w_{d}\right) g||_1.
	\end{split}
\end{equation}
Moreover, let us denote the recovered data as
\begin{equation}
 w^* = \mathscr{H}_L(w_{d}) g^*.   
\end{equation}
Notice that by imposing the $\ell_1$-norm, we aim to achieve a convex relaxation for the $0$-norm, specifically $||w- w^*||_0$. This encourages sparsity in $w- w^*$, which finds the best fit of the collected data at the most entries. 

\revise{While \eqref{eqn:lasso} is an approximation of the original problem~\eqref{eqn:P0}, we prove that under certain conditions, it is also able to exactly exclude the malicious behaviors and recover the real trajectory.}

\begin{theorem}\label{thm:lasso}
  Suppose that \eqref{eqn:PE} is satisfied.  Problem~\eqref{eqn:lasso} recovers the true trajectory. i.e., $w^* = \bar{w}$, if for any $v\in\mathbb{R}^{T-L+1}$ and $v\neq 0$, it holds 
    \begin{equation}\label{eqn:sufficient1}
        ||\mathscr{H}_{\text{B}} v||_1 > ||\mathscr{H}_{\text{F}} v||_1,
    \end{equation}
    where \begin{equation}\label{eqn:h1}
        \mathscr{H}_{\text{B}} \triangleq \mathscr{H}|_{[qL]\backslash \mathcal{C}},\; \mathscr{H}_{\text{F}}  \triangleq \mathscr{H}|_{\mathcal{C}}.
    \end{equation}
\end{theorem}
\begin{proof}
  Without loss of generality, let us assume
  \begin{equation}\mH = 
      \begin{bmatrix}
          \mathscr{H}_{\text{B}}\\
          \mathscr{H}_{\text{F}} 
      \end{bmatrix}.
  \end{equation} Correspondingly, $w$ is partitioned as 
  \begin{equation} w=
      \begin{bmatrix}
          w_{\text{B}}\\
          w_{\text{F}} 
      \end{bmatrix}=\begin{bmatrix}
          \bar{w}_{\text{B}}\\
          w_{\text{F}} 
      \end{bmatrix},
  \end{equation} where $w_{\text{B}} \triangleq w|_{[qL]\backslash \mathcal{C}},\; w_{\text{F}}  \triangleq w|_{\mathcal{C}}$, and $\bar{w}_{\text{B}} \triangleq \bar w|_{[qL]\backslash \mathcal{C}}$. To ensure $w^* = \bar{w}$, it is sufficient to show that $\bar{g}$ is the optimal solution of \eqref{eqn:lasso}, i.e., $g^*=\bar{g}$, where $\bar{g}$ generates the true trajectory as defined in \eqref{eqn:gbar}. This holds if and only if for any $v\neq 0$, \begin{equation}
        ||w - \mathscr{H} (\bar{g}+ v)||_1 > ||w - \mathscr{H} \bar{g}||_1.
    \end{equation}
   Since $\bar{w}_{\text{B}}=\mathscr{H}_{\text{B}} \bar{g}$, the above inequality is equivalent to
   \begin{equation}
       \left|\left|\begin{bmatrix}
          \mathscr{H}_{\text{B}} v\\
         w_{\text{F}}  - \mathscr{H}_{\text{F}}  (\bar{g} + v) 
       \end{bmatrix}
 \right|\right|_1  >||w_{\text{F}} - \mathscr{H}_{\text{F}}  \bar{g}||_1.
   \end{equation} Or, equivalently,
   \begin{equation}
       ||\mathscr{H}_{\text{B}} v||_1>||w_{\text{F}} - \mathscr{H}_{\text{F}}  \bar{g}||_1 - ||w_{\text{F}}  - \mathscr{H}_{\text{F}}  (\bar{g} + v) ||_1.
   \end{equation}
  The above inequality follows under \eqref{eqn:sufficient1}, which completes the proof.
\end{proof}

\revise{
\begin{remark}
	The main difficulty in using Theorem~\ref{thm:lasso} is to validate \eqref{eqn:sufficient1} over all nonzero $v$'s. In this remark, we provide a way to verify the condition. Let us define the following epigraphs:
	\begin{equation*}
		\mathcal{E}(B) \triangleq\left\{v: ||\mathscr{H}_{\text{B}} v||_1 \leq 1\right\}, \; 	\mathcal{E}(F) \triangleq\left\{v: ||\mathscr{H}_{\text{F}} v||_1 \leq 1\right\}
	\end{equation*}
	We then propose the following lemma:
	\begin{lemma}
		If there exists $\delta>0$ such that
		\begin{equation}\label{eqn:R1}
			\mathcal{E}(B) \subseteq (1-\delta) \mathcal{E}(F),  
		\end{equation}
		then \eqref{eqn:sufficient1} holds. 
	\end{lemma}
	\begin{proof}
		Let us consider any $v\neq 0$. Since \eqref{eqn:R1} holds, there must exist $\alpha \neq 0$ and $v_0=\alpha v$ such that $v_0 \in  \mathcal{E}(F) \backslash\mathcal{E}(B)$. That is, $||\mathscr{H}_{\text{B}} v_0||_1>1$ and $||\mathscr{H}_{\text{F}} v_0||_1\leq 1$. Therefore, it follows $||\mathscr{H}_{\text{B}} v_0||_1>||\mathscr{H}_{\text{F}} v_0||_1$. Then it is easy to check $||\mathscr{H}_{\text{B}} v||_1=\alpha||\mathscr{H}_{\text{B}} v_0||_1>\alpha ||\mathscr{H}_{\text{F}} v_0||_1 = ||\mathscr{H}_{\text{F}} v||_1$.
	\end{proof}
	Therefore, it is sufficient to verify \eqref{eqn:R1}. Since $\mathcal{E}(B)$ and $\mathcal{E}(F)$ are both convex polytopes \cite[Chapter~3]{boyd2004convex}, we only need to check
	whether all the vertices of $\mathcal{E}(B)$ lie in the interior of $\mathcal{E}(F)$. The vertices of $\mathcal{E}(B)$ can be found by solving the equation $||\mathscr{H}_{\text{B}} \nu||_1 = 1$, which can be done using linear programming. If for every vertex $\nu$, $||\mathscr{H}_{\text{F}} \nu||_1\leq 1$ holds, then \eqref{eqn:R1} is verified. Alternatively, one can also solve the following optimization problem
	\begin{equation*}
		\begin{split}
			\max_{\nu} \quad &||\mathscr{H}_{\text{F}} \nu||_1,\\
			\text{s.t. } \quad & ||\mathscr{H}_{\text{B}} \nu||_1 = 1.
		\end{split}
	\end{equation*}
	If the optimal value of the above problem is less than $1$, we are also able to validate \eqref{eqn:R1} and thus \eqref{eqn:sufficient1}. 		
\end{remark}}

Instead of \eqref{eqn:sufficient1}, the following condition can also guarantee the effectiveness of \eqref{eqn:lasso}.
\begin{condition}\label{con:H}
 $\rank \mathscr{H}_{\text{B}} = \rank \mathscr{H}$, where $\mathscr{H}_{\text{B}}$ is defined in~\eqref{eqn:h1}. 
\end{condition}
That is, removing ``malicious rows'' does not harm the rank of $\mathscr{H}$. \revise{Note that if Condition~\ref{con:H} does not hold, the benign data itself is not persistently exciting and therefore cannot fully determine the true solution set. This gives the adversary some degree of freedom to arbitrarily modify the data without being detected.}

Under Condition~\ref{con:H}, there exists a matrix $T$ such that
\begin{equation}\label{eqn:T}
    \mathscr{H}_{\text{F}}  = T \mathscr{H}_{\text{B}}.
\end{equation}
\begin{theorem}\label{prop:lasso}
  Suppose that \eqref{eqn:PE} is satisfied and Condition~\ref{con:H} holds. Problem~\eqref{eqn:lasso} recovers the true trajectory. i.e., $w^* = \bar{w}$, if $||T||_1 <1$.
\end{theorem}
\begin{proof}
    From \eqref{eqn:T}, we have $||\mathscr{H}_{\text{F}}  v||_1 = ||T\mathscr{H}_{\text{B}} v||_1\leq ||T||_1||\mathscr{H}_{\text{B}} v||_1<||\mathscr{H}_{\text{B}} v||_1$ for any $v\neq 0$. By using Theorem~\ref{thm:lasso}, we finish the proof.
\end{proof}

In general, the matrix $T$ that solves \eqref{eqn:T} is not unique. If there exists a matrix $T$ among all possible solutions that satisfies $||T||_1 < 1$, then Theorem~\ref{prop:lasso} holds.

Note also that no prior knowledge on the number of either false entries or channels is required for solving \eqref{eqn:lasso}, making it applicable to both entry-attacked and channel-attacked scenarios\footnote{In the channel-attacked scenario, instead of \eqref{eqn:h1}, we define $\mathscr{H}_{\text{B}} \triangleq \mathscr{H}|_{[qL]\backslash \mathscr{C}^{q,L}(\mathcal{C})}$ and $\mathscr{H}_{\text{F}} \triangleq \mathscr{H}|_{\mathscr{C}^{q,L}(\mathcal{C})}$. All results remain valid.}. Moreover, the channel-attacked scenario can also be approximated by using the idea of group LASSO \cite{meier2008group}:
\begin{equation}\label{eqn:grp_lasso}
	\begin{split}
		\min_{g, \beta} \quad &\sum_{i=1}^{q} ||\beta|_{\mathscr{C}_i^{q, L}}||_2\\
		s.t. \quad &\beta = w - \mathscr{H}_L\left(w_{d}\right) g,
	\end{split}
\end{equation}
where $\mathscr{C}_i^{q, L}$ is defined in \eqref{eqn:Ci}. We partition the entries of $\beta$ into $q$ groups based on $\mathscr{C}_i^{q, L},$ $i=1,\cdots,q$.
By using \eqref{eqn:grp_lasso}, we require that for each $\mathscr{C}_i^{q, L}$, either all entries of $\beta|_{\mathscr{C}_i^{q, L}}$ are zero or all of them are non-zero. Therefore, data collected from any channel is either entirely discarded or entirely maintained, which aligns with our assumption that the set of manipulated channels keeps constant over time.

\revise{\begin{remark}
Many algorithms are available in the literature to solve the $l_1$ norm minimization problem \eqref{eqn:lasso} and the group LASSO problem \eqref{eqn:grp_lasso},  including coordinate descent, exact path-following, and block coordinate gradient descent methods~\cite{buhlmann2011statistics}. The computational complexities of these algorithms are typically polynomial with respective to the number of channels, the size of the Hankel matrix, and the length of the trajectory. In contrast, brute-force algorithms show exponential complexity. Therefore, \eqref{eqn:lasso} and \eqref{eqn:grp_lasso} significantly reduce the computation costs.
\end{remark}}

\subsection{Extension to noisy data}\label{sec:noise}
In the presence of noisy data, we propose an approximation solution by solving a bi-level optimization problem.

\begin{enumerate}
    \item First, solving \eqref{eqn:lasso} yields the optimal solution $g^*$. 
    \item For each $w_i$, namely, the $i$-th entry of $w$, calculate the residual $|w_i-\mathscr{H}|_{\{i\}} g^*|$. Denote $\mathcal{C}^*$ as the index set of the largest $k$ residuals. 
    \item Let $\mI^* \triangleq [qL]/\mathcal{C}^*$. We solve the following least square problem to find the recovered trajectory $\tilde{w}$:
\begin{equation}
	\begin{split}
		\min_{g, \tilde{w}} \quad &||w|_{\mI^*}-\tilde{w}|_{\mI^*}||_2\\
		s.t. \quad &\tilde{w} = \mathscr{H}_L\left(w_{d}\right) g.
	\end{split}
\end{equation}
\end{enumerate}
The above procedure finds the recovered trajectory when up to $k$ entries of $w$ are subject to manipulation. For the channel-attacked scenario, we can have a similar procedure by leveraging \eqref{eqn:grp_lasso}. The details are omitted in the interest of space. In Section~\ref{sec:sim}, we will test the
effectiveness of these approximations through numerical examples.

\section{Numerical Examples}\label{sec:sim}
This section provides numerical examples to verify the theoretical results. The codes are available on \href{https://github.com/Yanjiaqi0328/Secure-Data-Reconstruction-A-Direct-Data-Driven-Approach}{Github}.

First, we consider the same example as in \cite{anand2023data}, see Fig.~\ref{fig:mass}. The displacement of the masses $m_1, m_2, m_3$ are denoted by $d_1, d_2, d_3$, respectively. A force $u$ is applied on mass $m_1$. The displacements of masses $m_2$ and $m_3$, as well as the velocity of mass $m_3$ are measured.
Define the state as \[x\triangleq [d_1, \dot{d}_1, d_2, \dot{d}_2, d_3, \dot{d}_3]^{\top}.\] Using Newton's laws of motion, a state-space representation of the system is established as
\begin{equation}\label{eqn:mass}
   \dot{x} = Ax+Bu, \quad y = Cx 
\end{equation}
where \begin{equation*}
\begin{split}
  A &= \begin{bmatrix}
    0 & 1 & 0 & 0 & 0 & 0 \\
-\frac{k_1}{m_1} & -\frac{b_1}{m_1} & 0 & 0 & 0 & 0 \\
0 & 0 & 0 & 1 & 0 & 0 \\
\frac{k_2}{m_2} & 0 & -\frac{k_2}{m_2} & -\frac{b_2}{m_2} & 0 & 0 \\
0 & 0 & 0 & 0 & 0 & 1 \\
0 & 0 & \frac{k_3}{m_3} & 0 & -\frac{k_3}{m_3} & -\frac{b_3}{m_3}   
   \end{bmatrix}, \\  
   B&=\begin{bmatrix}
   0 &
\frac{1}{m_1} &
0 &
0 &
0 &
0    
   \end{bmatrix}^{\top}, \\ C &= \bmx 0 & 0 & 1 & 0 & 0 & 0 \\ 0 & 0 & 0 & 0 & 1 & 0 \\ 0 & 0 & 0 & 0 & 0 & 1 \emx.
\end{split}
\end{equation*} 
The parameters in the model are
\revise{$k_1 = 2$N/m}, $m_1 = 1$kg, $b_1 = 3$Ns/m, $k_2 = 3$N/m, $m_2 = 2$kg, $b_2 = 4$Ns/m, $k_3 = 1$N/m,
$m_3 = 10$kg, and $b_3 = 2$Ns/m. We \revise{discretize} the model with a sampling time $T_s = 1.3$s. 
In this example, the system $\Beh$ is represented by a state-space model. We assume the system matrices are unknown, and only input/output data is collected. Specifically, $\invm(\Beh)=1$ and $\invp(\Beh)=3$ are the number of inputs and outputs, respectively.

\begin{figure}
    \centering
    \includegraphics[width=0.45\textwidth]{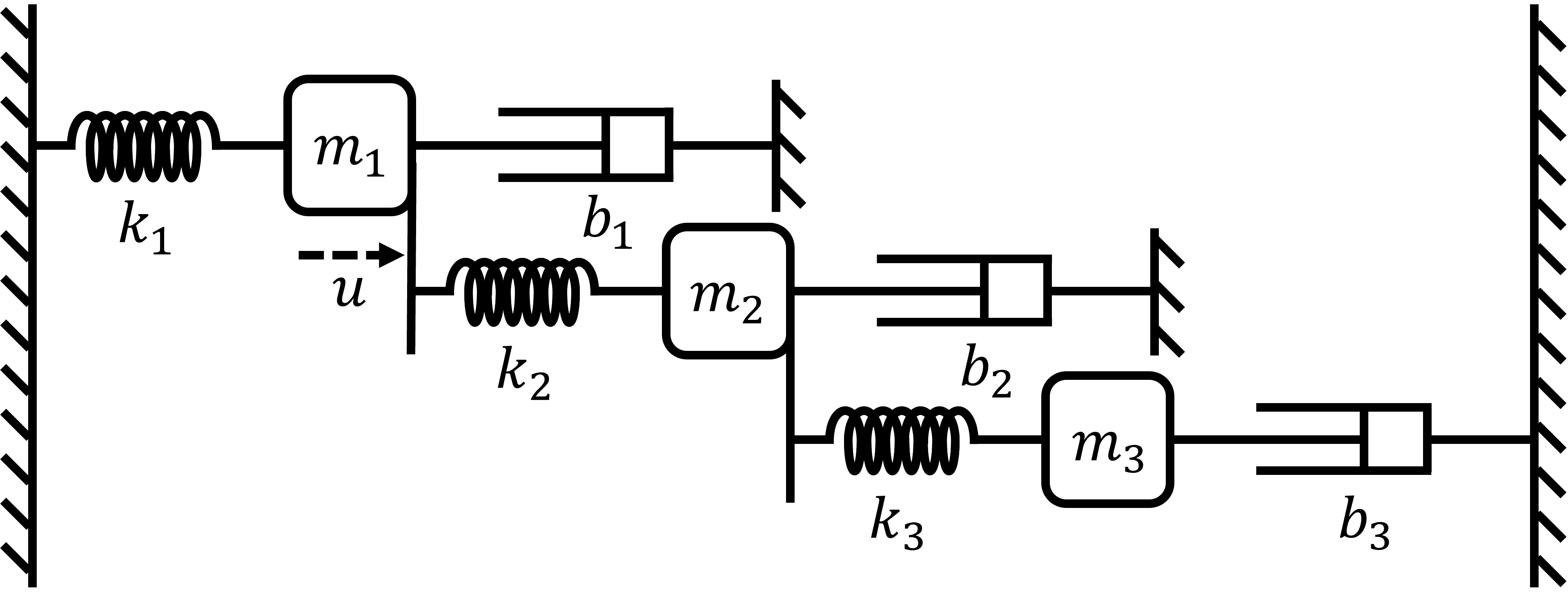}
    \caption{Three interconnected mass spring damper model.}
    \label{fig:mass}
\end{figure}

We generate $w_d$ with the length $T=11$, where the inputs are random variables drawn from a Gaussian distribution $\mathcal{N}(0,1)$ and verify that the generalized persistent condition is satisfied.


The true trajectory is partitioned into a sequence of data segments with length $L=3$. Moreover, the trajectory is contaminated by random noises drawn from $\mathcal{N}(0,0.5)$. Suppose that each segment can be manipulated by adversaries.
We assume $k=1$ and consider two scenarios:
\begin{enumerate}
    \item In the entry-attacked scenario, a single entry of each segment is modified at each step.
    \item In the channel-attacked scenario, all data from the second output channel are manipulated.
\end{enumerate}

\revise{Condition~\ref{con:H} holds in this example, which can be easily verified \textit{offline} since \(\mathscr{H}_{L}(w_d)\) is provided.} During the online computation, we solve the bi-level optimization problem proposed in Section~\ref{sec:noise} by leveraging \eqref{eqn:lasso} and \eqref{eqn:grp_lasso} respectively in the two scenarios, aiming to recover the real trajectory for each segment. As shown in Fig.~\ref{fig:l1} and Fig.~\ref{fig:grp_lasso}, both algorithms reconstruct the true trajectory in the presence of malicious behaviors and random noises. As compared with Fig.~\ref{fig:l1}, Fig.~\ref{fig:grp_lasso} yields worse performance since under the channel-attacked case, more data is manipulated. 
Note that in the figures, we only illustrate output trajectories at the first two dimensions, i.e., $d_2$ and $d_3$. The third output ($\dot{d}_3$) and the input ($u$), although omitted, are also successfully recovered.

\begin{figure}
    \centering
    \includegraphics[width=0.4\textwidth]{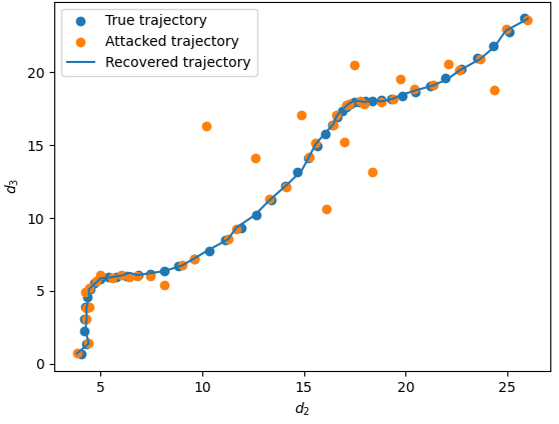}
    \caption{Trajectories under the entry-attacked model by using~\eqref{eqn:lasso}.}
    \label{fig:l1}
\end{figure}

\begin{figure}
    \centering
    \includegraphics[width=0.4\textwidth]{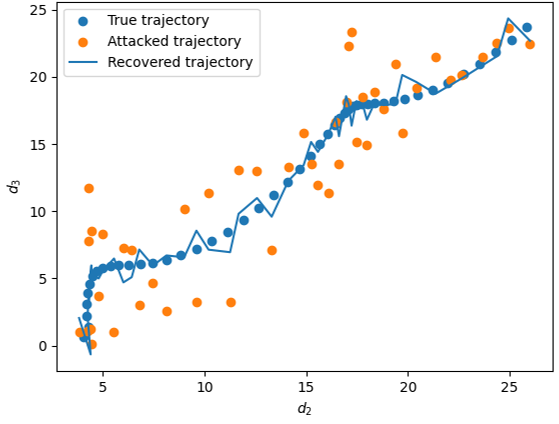}
    \caption{Trajectories under the channel-attacked model by using~\eqref{eqn:grp_lasso}.}
    \label{fig:grp_lasso}
\end{figure}

{\color{blue} \subsection{Comparison to other solutions}
we have added a comparison to the brute-force algorithm proposed in this paper, namely Algorithm~\ref{alg:combinatorial}. Furthermore, since there are few works in the literature that address the secure data reconstruction problem in unknown systems as we do, we also compare our method to a \textit{model-based} state reconstruction algorithm from [Mao2022], where the system matrices are known. 
	
	Notice that both Algorithm~\ref{alg:combinatorial} and [Mao2022] work in a noise-free environment. To ensure a fair comparison, we also apply our proposed solution \eqref{eqn:lasso} under noise-free conditions. The results are given in in Fig.~\ref{fig:compare}, which validate that \eqref{eqn:lasso} indeed recovers the exact trajectory. On the other hand, in this example with $3$ interconnected mass-spring-damper models, the computation times of Algorithm~\ref{alg:combinatorial} and \eqref{eqn:lasso} are $155.80$ms and $14.64$ms, respectively. That means \eqref{eqn:lasso} significantly reduces the complexity. It is clear that this advantage will become more evident as the problem size or the number of malicious entries increases. Moreover, different from [Mao2022], it is important to note that our method does not rely on any prior model knowledge and is purely data-driven.

	\begin{figure*}[htbp]
		\centering
		\subfigure[Trajectories by using Algorithm~\ref{alg:combinatorial}.] {\includegraphics[width=.3\textwidth]{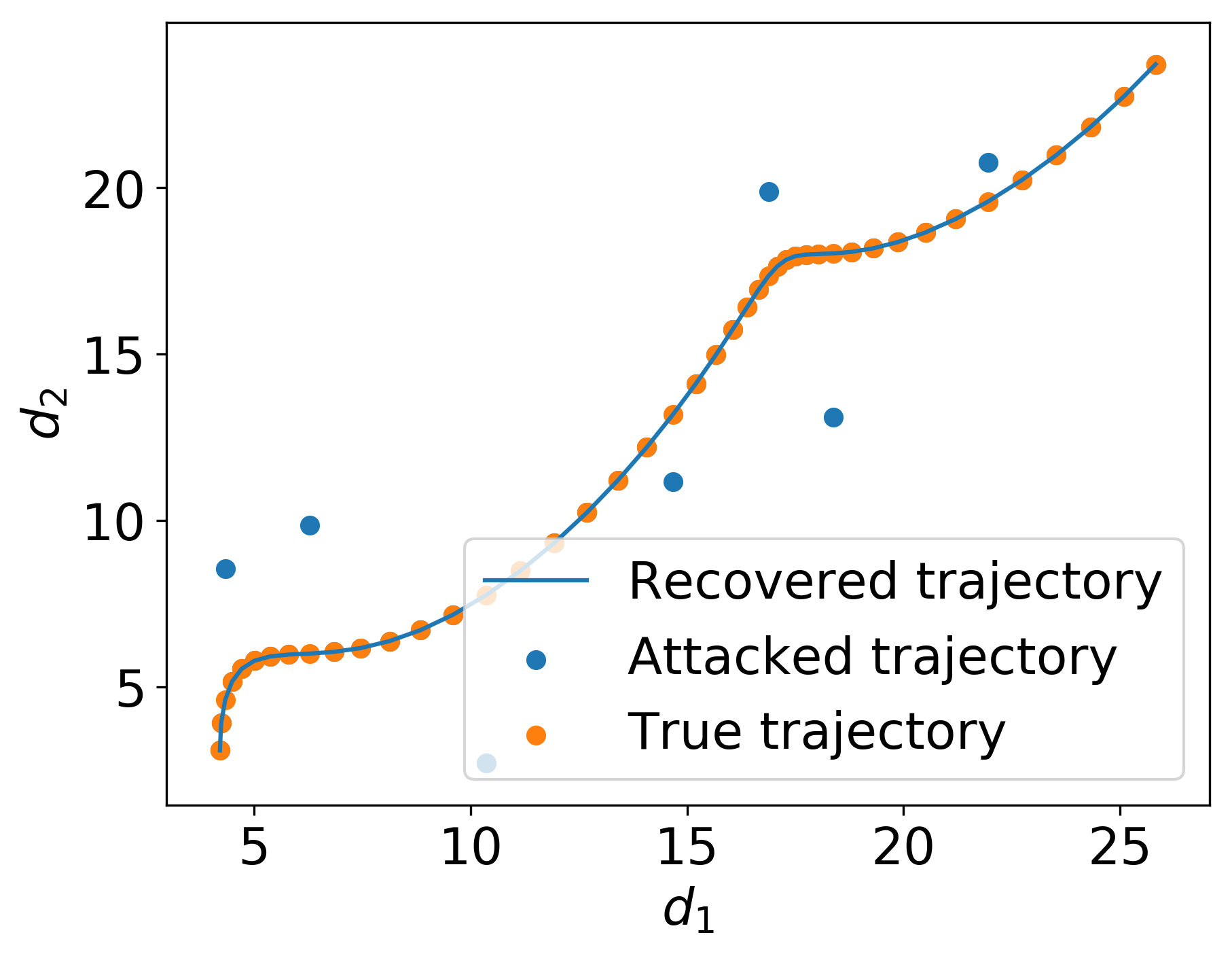}}
		\subfigure[Trajectories by using the model-based algorithm.] {\includegraphics[width=.3\textwidth]{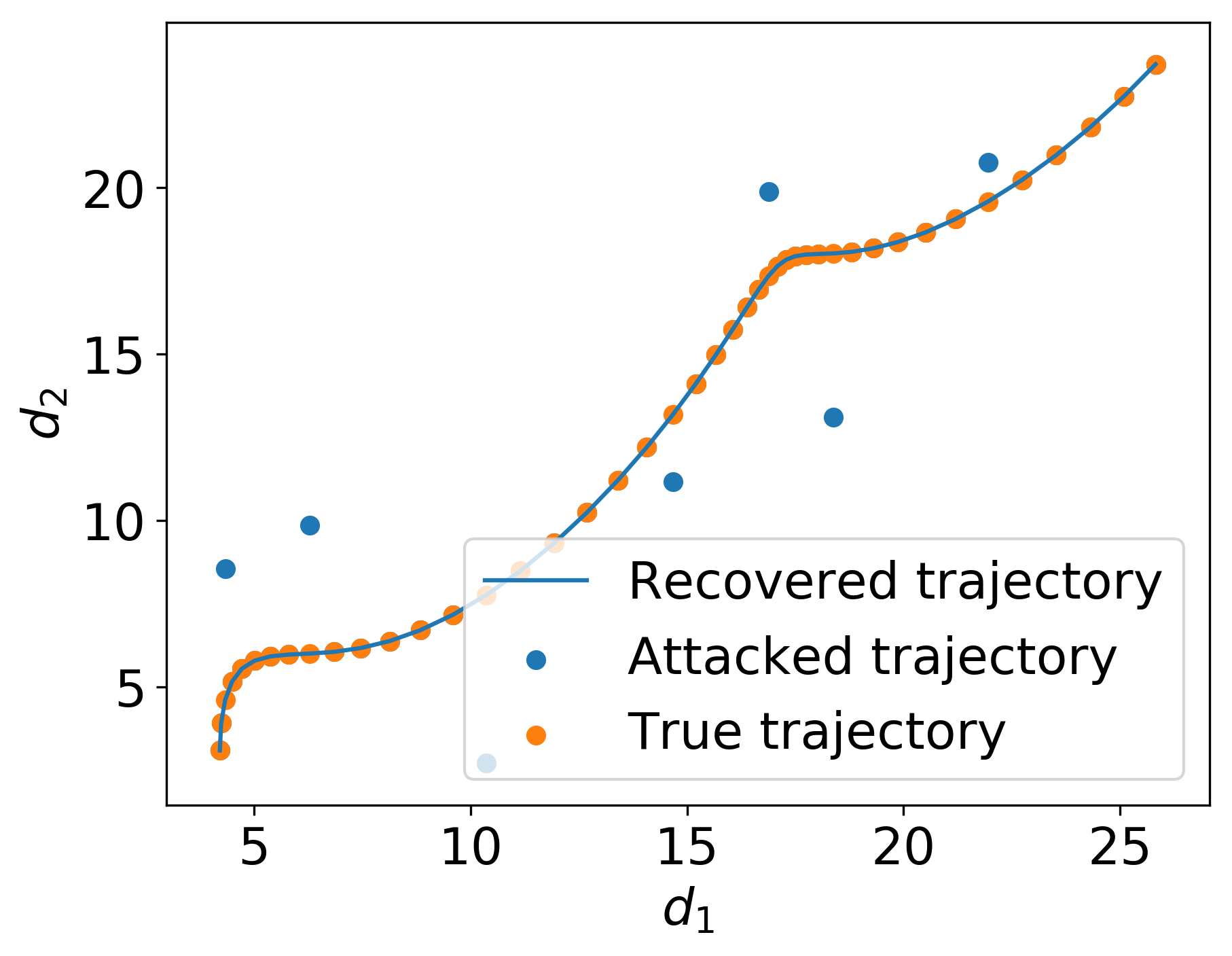}}
		\subfigure[Trajectories by using \eqref{eqn:lasso}.] {\includegraphics[width=.3\textwidth]{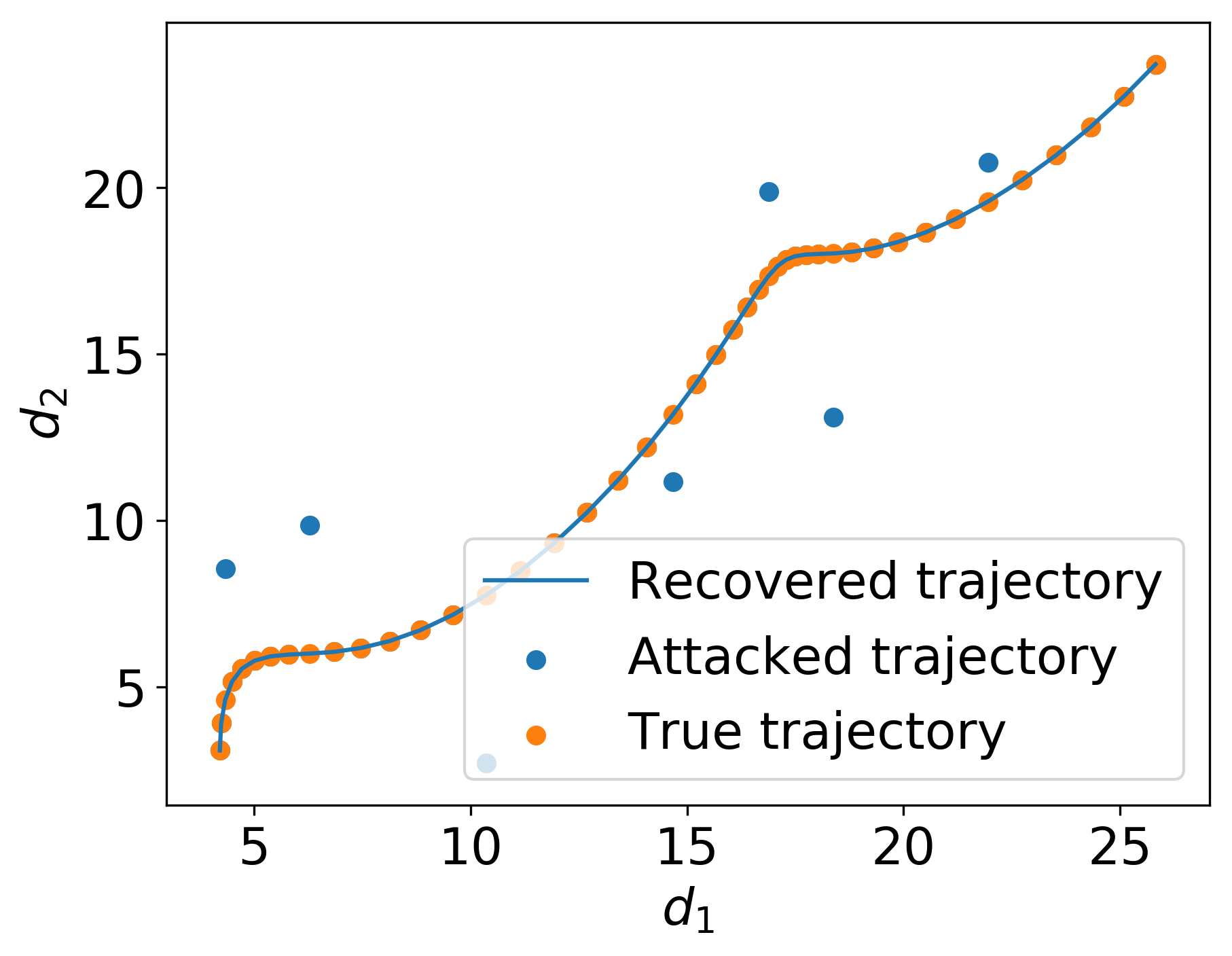}}
		\caption{\revise{Performance comparison among different algorithms.}}
		\label{fig:compare}
	\end{figure*}
	
	\subsection{Computation times in large-scale systems}
To evaluate the computational complexity of the methods proposed in Section~\ref{sec:solve}, we scale up \eqref{eqn:mass} to larger systems of $n$ interconnected mass-spring-damper models, where $n=\{3,10,20,30\}$. Let the state be the displacement and velocity of each mass, the input be a force applied on mass $m_1$, and the output be the displacements of masses $m_i, i\in\{2,3,\cdots,n\}$, as well as the velocity of $m_n$. The system dynamics can be obtained similarly by using Newton's law, with the state, input, and output dimensions being $2n$, $1$, and $n$, respectively. 

All computations are performed on a laptop with a 1.4GHz Intel(R) processor, and the optimization problems are solved using the CPLEX optimizer. We run each example for $50$ times. The computation times are recorded as below.
\begin{table}[!htbp]
	\centering
	\caption{\revise{Computation time of \eqref{eqn:lasso} on systems with different size.}}\label{tab:common}
	\begin{tabular}{c>{\centering}p{6mm}ccc}
		\toprule
		{\diagbox[dir=NW,innerwidth=4cm,height=2.5\line]{{\small Computation time (ms)}}{{\small No. interconnected}\\{\small  models}}} &{\small  $3$} & {\small $10$} & {\small $20$} & {\small $30$}\\
		\midrule
		{\small Average Computation Time} & {\small 4.24} & {\small 5.34} & {\small 7.92} & {\small 14.12} \\
		{\small Worst Computation Time} & {\small 21.72} & {\small 24.53} & {\small 24.80} & {\small 31.99}\\
		\bottomrule
	\end{tabular}
\end{table}

It is evident that the proposed approximate solutions scale well with respect to the problem size. 

Moreover, we show the performance in $n=10$, where the first input channel and the second output channel are attacked simultaneously. The data recovery errors for all channels are shown in Figure~\ref{fig:all_channel}. It is clear that the data on each channel can be successfully recovered.

\begin{figure}[!ht]
	\centering
	\includegraphics[width=0.4\textwidth]{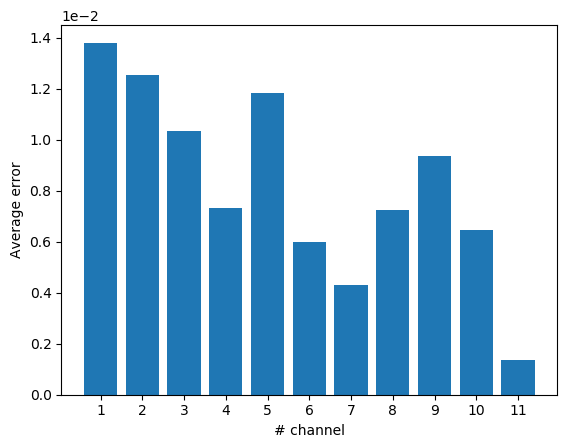}
	\caption{\revise{The average error of each channel.}}
	\label{fig:all_channel}
\end{figure}

\subsection{Performance on a quadruple-tank system}
To demonstrate the applicability of the proposed method across different systems, we further test \eqref{eqn:lasso} on a quadruple-tank system using the same parameters as \cite{zhang2023dimension}. The noise is modeled with a uniform distribution, and the second input channel is subjected to an attack. As shown in Fig.~\ref{fig:tank}, the proposed method also successfully recovers the true trajectory in this scenario.

\begin{figure}[!ht]
	\centering
	\includegraphics[width=0.45\textwidth]{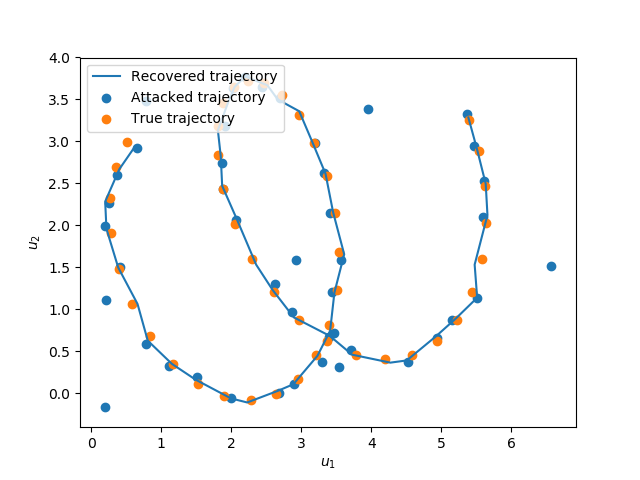}
	\caption{\revise{Trajectories under the channel-attacked model in the quadruple-tank system.}}
	\label{fig:tank}
\end{figure}

}
\section{Conclusion}\label{sec:cond}
This paper addresses the problem of trajectory reconstruction under malicious behaviors. The system is unknown, and we develop direct data-driven approaches without identifying the explicit model. Both entry-attacked and channel-attacked scenarios are investigated. With the behavioral language, we describe the system using trajectories and formulate the problem as optimization problems for both scenarios. Sufficient conditions are presented to guarantee successful data reconstruction. Moreover, we develop approximation solutions to mitigate computation costs and handle noisy data. We further demonstrate that under certain conditions, these approximation solutions also accurately yield the true trajectory.

\revise{For the future works, it is possible to extend the results presented in the paper to nonlinear systems or time-varying systems. Note that the fundamental lemma has already been generalized to various classes of nonlinear systems, such as Hammerstein–Wiener systems \cite{berberich2020}, second-order discrete Volterra systems \cite{9304122}, flat nonlinear systems \cite{9683327}, bilinear systems \cite{ddsim-narx}, and time-varying systems \cite{9683151}. These advancements may enable the achievement of CPS resiliency for more complex systems through direct data-driven methods.}

  \bibliographystyle{IEEEtran}
  \bibliography{reference}

\begin{thebibliography}{10}
\providecommand{\url}[1]{#1}
\csname url@samestyle\endcsname
\providecommand{\newblock}{\relax}
\providecommand{\bibinfo}[2]{#2}
\providecommand{\BIBentrySTDinterwordspacing}{\spaceskip=0pt\relax}
\providecommand{\BIBentryALTinterwordstretchfactor}{4}
\providecommand{\BIBentryALTinterwordspacing}{\spaceskip=\fontdimen2\font plus
\BIBentryALTinterwordstretchfactor\fontdimen3\font minus \fontdimen4\font\relax}
\providecommand{\BIBforeignlanguage}[2]{{%
\expandafter\ifx\csname l@#1\endcsname\relax
\typeout{** WARNING: IEEEtran.bst: No hyphenation pattern has been}%
\typeout{** loaded for the language `#1'. Using the pattern for}%
\typeout{** the default language instead.}%
\else
\language=\csname l@#1\endcsname
\fi
#2}}
\providecommand{\BIBdecl}{\relax}
\BIBdecl

\bibitem{cardenas2008secure}
A.~A. Cardenas, S.~Amin, and S.~Sastry, ``Secure control: Towards survivable cyber-physical systems,'' in \emph{2008 The 28th International Conference on Distributed Computing Systems Workshops}.\hskip 1em plus 0.5em minus 0.4em\relax IEEE, 2008, pp. 495--500.

\bibitem{chowdhury2020observer}
N.~R. Chowdhury, J.~Belikov, D.~Baimel, and Y.~Levron, ``Observer-based detection and identification of sensor attacks in networked cpss,'' \emph{Automatica}, vol. 121, p. 109166, 2020.

\bibitem{teixeira2010networked}
A.~Teixeira, H.~Sandberg, and K.~H. Johansson, ``Networked control systems under cyber attacks with applications to power networks,'' in \emph{Proceedings of the 2010 American control conference}.\hskip 1em plus 0.5em minus 0.4em\relax IEEE, 2010, pp. 3690--3696.

\bibitem{deng2016false}
R.~Deng, G.~Xiao, R.~Lu, H.~Liang, and A.~V. Vasilakos, ``False data injection on state estimation in power systems—attacks, impacts, and defense: A survey,'' \emph{IEEE Transactions on Industrial Informatics}, vol.~13, no.~2, pp. 411--423, 2016.

\bibitem{mishra2016secure}
S.~Mishra, Y.~Shoukry, N.~Karamchandani, S.~N. Diggavi, and P.~Tabuada, ``Secure state estimation against sensor attacks in the presence of noise,'' \emph{IEEE Transactions on Control of Network Systems}, vol.~4, no.~1, pp. 49--59, 2016.

\bibitem{mao2022computational}
Y.~Mao, A.~Mitra, S.~Sundaram, and P.~Tabuada, ``On the computational complexity of the secure state-reconstruction problem,'' \emph{Automatica}, vol. 136, p. 110083, 2022.

\bibitem{li2021low}
Z.~Li and Y.~Mo, ``Low complexity secure state estimation design for linear system with non-derogatory dynamics,'' in \emph{2021 60th IEEE Conference on Decision and Control (CDC)}.\hskip 1em plus 0.5em minus 0.4em\relax IEEE, 2021, pp. 6591--6596.

\bibitem{nakahira2018attack}
Y.~Nakahira and Y.~Mo, ``Attack-resilient $\mathcal{H}_2 $, $\mathcal{H}_\infty $, and $\ell_1 $ state estimator,'' \emph{IEEE Transactions on Automatic Control}, vol.~63, no.~12, pp. 4353--4360, 2018.

\bibitem{yan2022resilient}
J.~Yan, X.~Li, Y.~Mo, and C.~Wen, ``Resilient multi-dimensional consensus in adversarial environment,'' \emph{Automatica}, vol. 145, p. 110530, 2022.

\bibitem{leblanc2013resilient}
H.~J. LeBlanc, H.~Zhang, X.~Koutsoukos, and S.~Sundaram, ``Resilient asymptotic consensus in robust networks,'' \emph{IEEE Journal on Selected Areas in Communications}, vol.~31, no.~4, pp. 766--781, 2013.

\bibitem{piga2017direct}
D.~Piga, S.~Formentin, and A.~Bemporad, ``Direct data-driven control of constrained systems,'' \emph{IEEE Transactions on Control Systems Technology}, vol.~26, no.~4, pp. 1422--1429, 2017.

\bibitem{dorfler2022role}
F.~D{\"o}rfler, P.~Tesi, and C.~De~Persis, ``On the role of regularization in direct data-driven lqr control,'' in \emph{2022 IEEE 61st Conference on Decision and Control (CDC)}.\hskip 1em plus 0.5em minus 0.4em\relax IEEE, 2022, pp. 1091--1098.

\bibitem{markovsky2022data}
I.~Markovsky and F.~D{\"o}rfler, ``Data-driven dynamic interpolation and approximation,'' \emph{Automatica}, vol. 135, p. 110008, 2022.

\bibitem{markovsky2021behavioral}
------, ``Behavioral systems theory in data-driven analysis, signal processing, and control,'' \emph{Annual Reviews in Control}, vol.~52, pp. 42--64, 2021.

\bibitem{markovsky2023data}
I.~Markovsky, L.~Huang, and F.~D{\"o}rfler, ``Data-driven control based on the behavioral approach: From theory to applications in power systems,'' \emph{IEEE Control Systems Magazine}, vol.~43, no.~5, pp. 28--68, 2023.

\bibitem{willems2005note}
J.~C. Willems, P.~Rapisarda, I.~Markovsky, and B.~L. De~Moor, ``A note on persistency of excitation,'' \emph{Systems \& Control Letters}, vol.~54, no.~4, pp. 325--329, 2005.

\bibitem{markovsky2022identifiability}
I.~Markovsky and F.~D{\"o}rfler, ``Identifiability in the behavioral setting,'' \emph{IEEE Transactions on Automatic Control}, vol.~68, no.~3, pp. 1667--1677, 2022.

\bibitem{fl}
I.~Markovsky, E.~Prieto-Araujo, and F.~Dörfler, ``On the persistency of excitation,'' \emph{Automatica}, p. 110657, 2023.

\bibitem{berberich2020}
J.~Berberich and F.~Allgöwer, ``A trajectory-based framework for data-driven system analysis and control,'' in \emph{2020 European Control Conference (ECC)}, 2020, pp. 1365--1370.

\bibitem{9304122}
J.~G. Rueda-Escobedo and J.~Schiffer, ``Data-driven internal model control of second-order discrete volterra systems,'' in \emph{2020 59th IEEE Conference on Decision and Control (CDC)}, 2020, pp. 4572--4579.

\bibitem{9683327}
M.~Alsalti, J.~Berberich, V.~G. Lopez, F.~Allgöwer, and M.~A. Müller, ``Data-based system analysis and control of flat nonlinear systems,'' in \emph{2021 60th IEEE Conference on Decision and Control (CDC)}, 2021, pp. 1484--1489.

\bibitem{9683151}
C.~Verhoek, R.~Tóth, S.~Haesaert, and A.~Koch, ``Fundamental lemma for data-driven analysis of linear parameter-varying systems,'' in \emph{2021 60th IEEE Conference on Decision and Control (CDC)}, 2021, pp. 5040--5046.

\bibitem{ddsim-narx}
I.~Markovsky, ``Data-driven simulation of generalized bilinear systems via linear time-invariant embedding,'' vol.~68, pp. 1101--1106, 2023.

\bibitem{W07}
J.~C. Willems, ``The behavioral approach to open and interconnected systems: {M}odeling by tearing, zooming, and linking,'' \emph{Control Systems Magazine}, vol.~27, pp. 46--99, 2007.

\bibitem{mao2019secure}
Y.~Mao, A.~Mitra, S.~Sundaram, and P.~Tabuada, ``When is the secure state-reconstruction problem hard?'' in \emph{2019 IEEE 58th Conference on Decision and Control (CDC)}.\hskip 1em plus 0.5em minus 0.4em\relax IEEE, 2019, pp. 5368--5373.

\bibitem{shoukry2017secure}
Y.~Shoukry, P.~Nuzzo, A.~Puggelli, A.~L. Sangiovanni-Vincentelli, S.~A. Seshia, and P.~Tabuada, ``Secure state estimation for cyber-physical systems under sensor attacks: A satisfiability modulo theory approach,'' \emph{IEEE Transactions on Automatic Control}, vol.~62, no.~10, pp. 4917--4932, 2017.

\bibitem{ren2018binary}
X.~Ren, J.~Yan, and Y.~Mo, ``Binary hypothesis testing with byzantine sensors: Fundamental tradeoff between security and efficiency,'' \emph{IEEE Transactions on Signal Processing}, vol.~66, no.~6, pp. 1454--1468, 2018.

\bibitem{fawzi2014secure}
H.~Fawzi, P.~Tabuada, and S.~Diggavi, ``Secure estimation and control for cyber-physical systems under adversarial attacks,'' \emph{IEEE Transactions on Automatic control}, vol.~59, no.~6, pp. 1454--1467, 2014.

\bibitem{boyd2004convex}
S.~Boyd and L.~Vandenberghe, \emph{Convex optimization}.\hskip 1em plus 0.5em minus 0.4em\relax Cambridge university press, 2004.

\bibitem{meier2008group}
L.~Meier, S.~Van De~Geer, and P.~B{\"u}hlmann, ``The group lasso for logistic regression,'' \emph{Journal of the Royal Statistical Society Series B: Statistical Methodology}, vol.~70, no.~1, pp. 53--71, 2008.

\bibitem{buhlmann2011statistics}
P.~B{\"u}hlmann and S.~Van De~Geer, \emph{Statistics for high-dimensional data: methods, theory and applications}.\hskip 1em plus 0.5em minus 0.4em\relax Springer Science \& Business Media, 2011.

\bibitem{anand2023data}
S.~C. Anand, M.~S. Chong, and A.~M. Teixeira, ``Data-driven identification of attack-free sensors in networked control systems,'' \emph{arXiv preprint arXiv:2312.04845}, 2023.

\bibitem{zhang2023dimension}
K.~Zhang, Y.~Zheng, C.~Shang, and Z.~Li, ``Dimension reduction for efficient data-enabled predictive control,'' \emph{IEEE Control Systems Letters}, 2023.

\end{thebibliography}

  \end{document}